\documentclass[preprint,12pt]{elsarticle}

\usepackage{amssymb}
\usepackage{amsmath}
\usepackage{amsthm}
\newtheorem{proposition}{Proposition}
\usepackage{float}
\usepackage{amssymb}
\usepackage{url}
\usepackage{xcolor}

\journal{Physica A: Statistical Mechanics and its Applications}

\begin{document}

\begin{frontmatter}



\title{A Continuous Nonlinear Optimization Perspective on the Spin Glass Problem}

\author{Phil Duxbury}
\affiliation{
    organization={Michigan State University}, 
    addressline={Room 4260 Biomedical and Physical Science Building}, 
    city={East Lansing},
    postcode={48824-1116}, 
    state={MI},
    country={USA}
}

\author{Carlile Lavor}
\affiliation{
    organization={IMECC – UNICAMP, University of Campinas}, 
    city={Campinas},
    postcode={13083-790}, 
    state={SP},
    country={Brazil}
}

\author{Luiz Leduino de Salles-Neto}
\affiliation{
    organization={Institute of Science and Technology, Federal University of São Paulo}, 
    addressline={Avenida Cesare Mansueto Giulio Lattes, 1201}, 
    city={São José dos Campos},
    postcode={12247-014}, 
    state={SP},
    country={Brazil}
}


\begin{abstract}
We present a continuous nonlinear optimization model for the Spin Glass Problem (SGP), building on a classical result by Rosenberg (1972), which shows that for a class of multilinear polynomial problems the optimal values of the continuous relaxation and the corresponding discrete model coincide. Using the SGP as a case study, we provide a simple, problem-specific argument showing how any optimal solution returned by a continuous solver can be converted into an optimal discrete spin configuration, even when the solver outputs non-integer values. The relaxed model remains nonconvex and does not alter the inherent computational hardness of the problem, but it offers a direct and conceptually transparent continuous formulation that can be handled by modern global optimization software. Computational experiments on standard benchmark instances indicate that this approach can match, and in several cases surpass, recent integer programming linearization techniques, making it a practical and complementary tool for researchers working at the interface between statistical physics and combinatorial optimization.
\end{abstract}


\end{frontmatter}



\section{Introduction}

The \textit{Spin Glass Problem} (SGP) is a central problem in statistical physics and materials science, with additional applications in chemistry and artificial
neural networks \cite{bo07,fi91}. In its simplest form, the model describes a collection of interacting spins arranged on a lattice, where each spin can take one of
two orientations.

A spin can be seen as a vector of fixed magnitude and variable direction. In a crystal with cubic symmetry, where the atomic arrangement induces preferred axes, it is common to restrict attention to a dominant direction, allowing each spin to assume one of two possible states, $-1$ or $+1$. A configuration of the system assigns a spin value to each lattice site and determines how neighboring spins interact.

For any pair of adjacent sites $i,j$ in the lattice, the coefficient $J_{i,j}$ represents the strength and nature of the magnetic interaction. The energy of a configuration is given by
\[
F(x_1,\ldots,x_n)=-\sum_{(i,j)\in E} J_{i,j}\, x_i x_j,
\]
where $E$ denotes the set of interacting pairs. This formulation is naturally interpreted on a graph whose vertices correspond to lattice sites and whose edges represent spin–spin couplings. A configuration partitions the vertices into those with $x_i=+1$ and those with $x_i=-1$, and in the absence of an external field the SGP is equivalent to a Max-Cut problem \cite{re10}.

Minimizing $F$ is also equivalent to solving the following nonlinear binary optimization problem, which we refer to as the \textit{SGP Model}:
\begin{equation}
\begin{array}{c}
\begin{array}{rc}
\underset{x_i}{\min}\;\; & 
F_1(x_1,\ldots,x_n)= -\displaystyle\sum_{i=1}^{n-1}\sum_{j>i}^{n} J_{i,j} x_i x_j \\[6pt]
& x_i \in \{-1,1\},\quad i=1,\ldots,n.
\end{array}
\end{array}
\label{model_1}
\end{equation}
The SGP Model is a binary quadratic problem and inherits the well-known computational hardness of Max-Cut.

\bigskip

\section{Large-Scale Ising Machines}
\label{sec:ising-machines}

Over the last decade, a variety of large-scale “Ising machines’’ have been proposed as heuristic solvers for NP-hard combinatorial optimization problems \cite{mohseni}. The common feature of these approaches is that the Ising or spin-glass Hamiltonian is encoded into a physical or dynamical system whose evolution tends to reach low-energy configurations, instead of being solved directly as a mathematical programming model.

Coherent Ising machines (CIMs) \cite{inagaki2016coherent} implement analog optical networks that approximately minimize the Ising energy and can handle fully connected instances of moderate size. Simulated Bifurcation (SB) algorithms \cite{goto2021highperformance} emulate certain Hamiltonian dynamics in discrete or continuous time and are efficiently implemented on FPGA or GPU hardware, addressing very large spin systems. Mean-field annealing schemes \cite{das2025classical} and vector-spin methods such as the Vector Ising Spin Annealer (VISA) \cite{cummins2025vectorising} also follow this paradigm, replacing discrete spins by continuous variables whose dynamics explore the energy landscape before projection back to $\{-1,1\}$ configurations.

These methods typically achieve high-quality solutions within short computational or physical runtimes and can scale to large instance sizes. However, they are designed as physics-inspired heuristics: they do not provide general guarantees of global optimality and their performance often depends on specialized or high-performance computing infrastructures.

In the remainder of this paper, we pursue a complementary direction. We formulate the SGP as a continuous nonlinear optimization model that is closely related to its binary quadratic formulation \eqref{model_1}, and we solve both using a modern global optimization solver running on general-purpose CPU hardware. This setting retains the benefits of a rigorous mathematical programming framework and, when the solver is run until convergence, allows solutions to be certified as globally optimal, while remaining independent of any specific analog or physics-inspired hardware architecture.

\section{The Spin Glass Problem as a Binary Quadratic Problem}

The SGP can be cast as a Binary Quadratic Problem (BQP) \cite{ca08}. Up to a standard change of variables from $\{-1,1\}$ to $\{0,1\}$, a generic BQP can be written as
\begin{equation}
\begin{array}{c}
\begin{array}{rc}
\underset{x_i}{\min}\;\; & \displaystyle \sum_{i=1}^{n} \sum_{j=1}^{n} Q_{i,j}\, x_i x_j
\;+\; \sum_{i=1}^{n} L_i x_i \\[6pt]
& x_i \in \{0,1\},\quad i=1,\ldots,n,
\end{array}
\end{array}
\label{modelfurini1}
\end{equation}
where $Q \in \mathbb{R}^{n\times n}$ is a symmetric matrix and $L \in \mathbb{R}^n$.

A widely used strategy for solving a BQP is to reformulate it as a mixed-integer linear program (MILP) by introducing additional binary variables \cite{glo73}. Other linearization techniques (LTs) have been proposed that rely only on additional continuous variables \cite{glo74,glo75,she07}. 

In \cite{fu19}, the authors present a detailed study of the strength of the linear programming (LP) relaxations associated with four such LTs, and introduce a new formulation, called the Extended Linear Formulation. Besides several theoretical results, they report a computational comparison of these LTs on problems with and without linear constraints.

We briefly recall here the two strongest LTs without linear constraints discussed in \cite{fu19}: the formulation proposed by Glover \cite{glo74}, denoted GW model, and the one by Furini and Traversi \cite{fu19}, denoted FT model (the other two linearizations analyzed in \cite{fu19} are due to Glover \cite{glo75} and Sherali and Smith \cite{she07}).

The \textit{GW Model} is based on a standard linearization of the quadratic terms in a BQP. It is given by
\begin{equation}
\begin{array}{c}
\begin{array}{rc}
\underset{x_i,\;y_{ij}}{\min}\;\; & 
\displaystyle \sum_{i=1}^{n-1}\sum_{j>i}^{n} 2Q_{i,j}\, y_{ij}
\;+\; \sum_{i=1}^{n} L_i x_i \\[6pt]
& 
\begin{array}{c}

y_{ij} \leq x_{i} \text{ } \;\;\;\;\;\;\;\;\;\;\;\;\;\; (i,j=1,...,n, \; i<j)\\
y_{ij} \leq x_{j} \text{ } \;\;\;\;\;\;\;\;\;\;\;\;\;\; (i,j=1,...,n, \; i<j)\\
y_{ij} \geq x_{i} + x_{j} -1 \; (i,j=1,...,n, \; i<j)\\


\text{ } y_{ij} \in \{0,1\} \;\;\;\;\;\;\;\;\;\;  (i,j=1,...,n, \; i<j)\\

\text{ }x_{i} \in \{0, 1\} \text{ } \;\;\;\;\;\;\;\;\;\;\;\;\;\;\;\;\;\;\;\;\; (i=1,...,n). \\ 

\end{array}%
\end{array}%
\end{array}%
\label{modelfurini2}
\end{equation}

This formulation increases the size of the problem by adding $n(n-1)/2$ variables and $4n(n-1)/2$ constraints.

The \textit{FT Model} is given by
\begin{equation}
\begin{array}{c}
\begin{array}{rc}
\underset{x_i,\;z_{ij}^i,\;z_{ij}^j}{\min}\;\; &
\displaystyle \sum_{i=1}^{n}\sum_{j=i}^{n} Q_{i,j}
\;+\; \sum_{i=1}^{n} L_i x_i
\;-\; \sum_{i=1}^{n-1}\sum_{j>i}^{n} 2Q_{i,j}\,(z_{ij}^i + z_{ij}^j) \\[6pt]
& 
\begin{array}{c}

x_{i}+z_{ij}^i \leq 1 \text{ } \;\;\;\;\;\;\;\;\;\;\;\;\;\; (i,j=1,...,n, \; i<j)\\
x_{j}+z_{ij}^j \leq 1 \text{ } \;\;\;\;\;\;\;\;\;\;\;\;\;\; (i,j=1,...,n, \; i<j)\\
z_{ij}^i+z_{ij}^j \leq 1 \text{ } \;\;\;\;\;\;\;\;\;\;\;\;\;\; (i,j=1,...,n, \; i<j)\\
x_i+z_{ij}^i+z_{ij}^j \geq 1 \text{ } \;\;\;\;\;\; (i,j=1,...,n, \; i<j)\\
x_{j}+z_{ij}^i+z_{ij}^j \geq 1 \text{ } \;\;\;\;\;\; (i,j=1,...,n, \; i<j)\\


\text{ } z^i_{ij} \in \{0,1\} \text{ } \;\;\;\;\;\;\;\;\;\;\;\;\;\;\; (i,j=1,...,n, \; i<j)\\

\text{ } z^j_{ij} \in \{0,1\} \text{ } \;\;\;\;\;\;\;\;\;\;\;\;\;\;\;(i,j=1,...,n, \; i<j)\\

\text{ }x_{i} \in \{0, 1\}\text{ } \;\;\;\;\;\;\;\;\;\;\;\;\;\;\;\;\;\;\;\;\;\;\;\;\;\;\; (i=1,...,n). \\

\end{array}%
\end{array}%
\end{array}%
\label{modelfurini2a}
\end{equation}
This formulation increases the problem size by $n(n-1)$ additional variables and $5n(n-1)$ constraints.

In the next section, we propose a different strategy for tackling the SGP. Instead of first formulating it as an unconstrained binary quadratic problem and then relaxing the associated linearized model, we directly relax the original SGP Model \eqref{model_1}. We then show that any optimal solution of this continuous relaxation can be mapped to an optimal solution of the discrete problem. Other continuous formulations for spin systems, unrelated to BQP linearizations, can be found in \cite{bi00,ha20,si73}.

\section{The Spin Glass Problem as a Continuous Optimization Problem}

Motivated by the results obtained in \cite{du21} (see also \cite{ouzia}), where a continuous relaxation of a discrete model for the \textit{Golomb Ruler Problem} (GRP) \cite{dra09} was successfully used to solve challenging instances, we now adopt a similar viewpoint for the SGP.

We consider the following formulation, called the \textit{Continuous Spin Glass} (CSG) model, or simply the \textit{CSG Model}:
\begin{equation}
\begin{array}{c}
\begin{array}{rc}
\underset{x_i}{\min} & 
F_2(x_1,\ldots,x_n)= -\displaystyle\sum_{i=1}^{n-1}\sum_{j>i}^{n} J_{i,j} x_i x_j \\[6pt]
& x_{i} \in [-1, 1],\;\; i=1,\ldots,n.
\end{array}
\end{array}
\label{model_2}
\end{equation}

In \cite{rose}, Rosenberg considered a more general problem: minimizing a real polynomial $p(x_1,\ldots,x_n)$ that is linear in each variable over the box $[-1,1]^n$. He proved that the minimum value of this continuous problem coincides with the minimum of $p$ over the discrete set $\{-1,1\}^n$, and in particular that there always exists an optimal solution with all components in $\{-1,1\}$. Our model \eqref{model_2} is a special case of this framework, since $F_2$ is multilinear and the feasible region is exactly $[-1,1]^n$.

Rosenberg’s result is essentially existential: it guarantees equality of optimal values and the existence of discrete optimal solutions, but it does not address how the output of a generic continuous solver should be processed in practice. In particular, continuous global optimization solvers may return optimal (or near-optimal) points with some components strictly inside $(-1,1)$, and it is not immediately obvious how to turn such points into discrete spin configurations without losing optimality.

In this paper, we regard the SGP as a concrete instance of Rosenberg’s setting and use his result to justify the continuous formulation \eqref{model_2}. For completeness, and to keep the exposition self-contained for readers outside mathematical optimization, we also state a simple and problem-specific argument that explains the structure of optimal solutions of \eqref{model_2} and underlies our procedure for converting continuous solutions into discrete ones. The key observation is summarized in the following proposition, which relies only on the linearity of $F_2$ in each variable and should be viewed as a classical property rather than a new theoretical result.

\begin{proposition}[Optimality transfer to box endpoints]
Let $F_2(x_1,\ldots,x_n)=-\sum_{i=1}^{n-1}\sum_{j>i}^{n} J_{i,j} x_i x_j$ be defined on the box $[-1,1]^n$, and let $x^\star$ be a global minimizer. If $x_k^\star\in(-1,1)$ for some $k$, then
\[
F_2\bigl(x^\star|_{x_k=-1}\bigr)
=
F_2\bigl(x^\star|_{x_k=1}\bigr)
=
F_2(x^\star).
\]
\end{proposition}

\begin{proof}
Fix all coordinates except $x_k$ and consider the univariate function
\[
\phi_k(t)
=
F_2\bigl(x_1^\star,\dots,x_{k-1}^\star,t,x_{k+1}^\star,\dots,x_n^\star\bigr),
\quad t \in [-1,1].
\]
Since $F_2$ is linear in each variable, $\phi_k$ is an affine function of~$t$. As $x^\star$ is a global minimizer of $F_2$ on $[-1,1]^n$ and $x_k^\star \in (-1,1)$, the point $t = x_k^\star$ must be a global minimizer of $\phi_k$ on $[-1,1]$. An affine function on an interval can attain a (local or global) minimum at an interior point only if its slope is zero, \textit{i.e.}, $\phi^{\prime}_k(x_k^\star)=0$, which implies that $\phi_k$ is constant on $[-1,1]$. Hence
\[
\phi_k(-1)  = \phi_k(1) = \phi_k(x_k^\star),
\]
which is equivalent to
\[
F_2\bigl(x^\star|_{x_k=-1}\bigr)
=
F_2\bigl(x^\star|_{x_k=1}\bigr)
=
F_2(x^\star).
\]
Since $x^\star$ is a global minimizer, the two endpoints have the same minimal value and are therefore global minimizers as well.
\end{proof}

By repeatedly applying Proposition~1 to any coordinate $x_k^\star \in (-1,1)$, we can move each such component to either $-1$ or $+1$ without changing the objective value. This shows that there always exists a global minimizer of \eqref{model_2} in $\{-1,1\}^n$, and that an optimal solution with non-integer components can be systematically transformed into an optimal discrete spin configuration.

In the next section, we compare models \eqref{modelfurini2} and \eqref{modelfurini2a}, studied in \cite{fu19}, with the CSG Model \eqref{model_2}, using the GUROBI global optimization solver \cite{gurobi}.

\section{Computational Experiments}

For the computational experiments, we considered the same instance classes used in \cite{fu19} for the SGP Model \eqref{model_1}. These classes, denoted \textit{rudy} and \textit{ising}, belong to the Biq Mac library \cite{wi07}, a widely used collection of benchmark instances for Max-Cut and related binary quadratic problems (see also \cite{kri14, re10}).

For the \textit{rudy} instances (Table~\ref{rudy1}), we considered problems with $n = 60$, $n = 80$, and $n = 100$ variables. For the \textit{ising} instances (Table~\ref{ising1}), we used instances with $n = 100$, $n = 150$, $n = 250$, and $n = 300$. Further details on these classes can be found in \cite{wi07} and all instances are publicly available at\\

\url{https://biqmac.aau.at/biqmaclib.html}\\

All experiments for the models \eqref{modelfurini2} and \eqref{modelfurini2a}, as well as for the CSG Model \eqref{model_2}, were carried out on the NEOS server \cite{gro97}, using a Dell EMC PowerEdge R440 with the following configuration: two Intel Xeon Gold 6140 CPUs @ 2.3GHz (36 cores in total), 512 GB RAM, and 2$\times$1.5 TB SAS SSDs in RAID1, connected via 1 Gb/s Ethernet. This setup corresponds to a general-purpose CPU server and does not involve any specialized Ising hardware or physics-inspired accelerators. We used GUROBI 12.0 \cite{gurobi} with default parameter settings for all models. For each instance, the time limit was set to 1800 seconds. The codes and parameter files employed in our experiments are available at\\

\url{https://github.com/luizleduino/spinglassmodel}\\

In Tables~\ref{rudy1}, \ref{ising1}, and~\ref{gset}, entries in bold indicate the best objective value among the models tested with GUROBI for each instance. Entries marked with an asterisk (*) coincide with the known optimal value (column \textit{Optimum}). Tables~\ref{rudy1} and~\ref{ising1} report, for each problem, the optimal value, the best solution obtained from the linearization-based approach (column \textit{Best LT}) and its runtime (\textit{Time LT}), and the solution returned by the CSG Model \eqref{model_2} (column \textit{CSG Model}) together with its runtime (\textit{Time CSG}).

Tables~\ref{rudy1} and~\ref{ising1} summarize the behavior of the competing formulations on the \textit{rudy} and \textit{ising} classes, respectively. For all instances reported in these tables, the CSG Model attains the known optimal value, as indicated by the entries that match the \textit{Optimum} column. In contrast, the linearization-based models often fail to reach the optimum within the time limit, especially for the \textit{rudy} instances, where their best values do not reach the optimum in most cases.

A second salient feature is the difference in computational time. On the \textit{ising} instances, the CSG Model typically converges in a fraction of a second, whereas the linearization-based models frequently requires much longer runtimes (hundreds or thousands of seconds) and, in several cases, still does not match the optimal value. 

Taken together, these results indicate that the continuous formulation \eqref{model_2}, when solved by a modern global optimizer such as GUROBI, provides a competitive and often superior alternative to classical linearization techniques for SGP-related binary quadratic problems, within a standard mathematical-programming workflow based on general-purpose CPU resources.

The results reported in Table~\ref{gset} compare models \eqref{modelfurini2} and \eqref{modelfurini2a} with the CSG Model \eqref{model_2} on a set of ten instances from the Gset library. These instances have a large number of vertices and are widely used as challenging benchmarks for Max-Cut and related Ising problems (see, for example, \cite{ben,festa}). They can be downloaded from\\

\url{http://web.stanford.edu/~yyye/yyye/Gset/} \\

In this set of ten instances, the CSG Model attains the known optimal value in five cases. For the remaining instances, the worst relative gap between the value returned by the CSG Model and the optimum is $2.35\%$ (instance G22). In five of the ten instances, GUROBI did not find any feasible solution for the linearization-based models within the time limit, while the CSG Model always produced a feasible solution of good quality. It is also worth noting that, for instances G48 ($n = 3000$), G49 ($n = 3000$), and G57 ($n = 5000$), GUROBI finds the optimal solution using the CSG Model in less than one second.

It is important to emphasize that none of the mathematical models considered here uses the optimal value as a target or input. The solvers operate without any prior knowledge of the optimum, which is only employed a posteriori to evaluate solution quality. This is particularly relevant in practical scenarios, where the true optimal value is typically unknown.

\begin{table}[t]
\caption{Rudy instances: the best solutions obtained by each approach.}
\begin{center}
\begin{tabular}{|c|c|c|c|c|c|c|}
\hline\hline
Instance  & n & Optimum & Best LT & Time LT &  CSG Model & Time CSG\\
\hline
$g05\_60.0$ & 60 & -536 & -535 & 1800 & \textbf{-536*} & 122.55\\
\hline
$g05\_60.1$ & 60 & -532 & -529 & 1800 &  \textbf{-532*} & 130.12 \\
\hline
$g05\_60.2$ & 60 & -529 & \textbf{-529*} & 1800 &  \textbf{-529*} & 116.81 \\
\hline
$g05\_80.0$ & 80 & -929 & -924 & 1800 &  \textbf{-929*} & 1800 \\
\hline
$g05\_80.1$ & 80 & -941 & -938 & 1800 &  \textbf{-941*}& 425.15\\
\hline
$g05\_80.2$ & 80 & -934 & -932 & 1800&  \textbf{-934*} & 977.78 \\
\hline
$g05\_100.0$ & 100 & -1430 & -1415 & 1800 & \textbf{-1430*}  & 1800 \\
\hline
$g05\_100.1$ & 100 & -1425 & -1415 & 1800& \textbf{-1425*} & 1800\\
\hline
$g05\_100.2$ & 100 & -1432 & -1413 & 1800 & \textbf{-1432*} & 1800\\
\hline
\end{tabular}%
\end{center}
\label{rudy1}
\end{table}

\begin{table}[H]
\caption{Ising instances: the best solutions obtained by each approach.}
\begin{center}
\begin{tabular}{|c|c|c|c|c|c|c|}
\hline\hline
Instance  & n & Optimum & Best LT & Time LT &  CSG Model & Time CSG\\
\hline
$2.5-100\_5555$ & 100 & -2460049 & \textbf{-2460049*} & 50.717 & \textbf{-2460049*} & 0.026\\
\hline
$2.5-100\_6666$ & 100 & -2031217 & \textbf{-2031217*} & 47.271 & \textbf{-2031217*} & 0.027\\
\hline
$2.5-100\_7777$ & 100& -3363230 & \textbf{-3363230*} & 115.163 &\textbf{-3363230*} & 0.042\\
\hline
$2.5-150\_5555$ & 150 & -4363532 &\textbf{-4363532*}& 394.32 &\textbf{-4363532*} & 0.051\\
\hline
$2.5-150\_6666$ & 150 & -4057153 &\textbf{-4057153*} &609.43 & \textbf{-4057153*} &0.045\\
\hline
$2.5-150\_7777$ & 150 & -4243269 &-4242853 & 1800 & \textbf{-4243269*} & 0.054 \\
\hline
$2.5-250\_5555$ & 250 & -7919449 & -7879016 &1800  & \textbf{-7919449*} &0.056 \\
\hline
$2.5-250\_6666$ & 250 & -69255717 & -6919872 &1800  & \textbf{-6925717*} & 0.051  \\
\hline
$2.5-250\_7777$ & 250 & -6596797 & -6573556 & 1800 & \textbf{-6596797*} & 0.052  \\
\hline
$3.0-250\_5555$ & 250 & -7823791 & \textbf{-7823791*} &1800  &  \textbf{-7823791*} & 0.064 \\
\hline
$3.0-250\_6666$ & 250 & -6903351 & -6900760 & 1800 &  \textbf{-6903351*} & 0.052 \\
\hline
$3.0-250\_7777$ & 250 & -6418276 & -6366224 & 1800 & \textbf{-6418276*} & 0.060\\
\hline
$2.5-300\_5555$ & 300 & -8579363 & -8487227 & 1800 & \textbf{-8579363*} & 0.051\\
\hline
$2.5-300\_6666$ & 300 & -9102033 & -8914176 & 1800 & \textbf{-9102033*} &  0.051\\
\hline
$2.5-300\_7777$ & 300 & -8323804 & -8197766 & 1800 & \textbf{-8323804*} &  0.052 \\
\hline
$3.0-300\_5555$ & 300 & -8493173 & -8487227 & 1800 & \textbf{-8493173*} & 0.057   \\
\hline
$3.0-300\_6666$ & 300 & -8915110 & -8914176 & 1800 & \textbf{-8915110*} & 0.070  \\
\hline
$3.0-300\_7777$ & 300 & -8242904 & -8236400 & 1800 & \textbf{-8242904*} &  0.071  \\
\hline
\hline
\end{tabular}%
\end{center}
\label{ising1}
\end{table}

\begin{table}[H]
\caption{Gset instances: the best solutions obtained by each approach.}
\begin{center}
\begin{tabular}{|c|c|c|c|c|c|c|c|}
\hline\hline
Instance  & n & Optimum & Best LT & Time LT &  CSG Model & Time CSG & GAP CSG\\
\hline
G1 & 800 & -11624 & -11433 & 1800 & \textbf{-11624*} & 1800 & 0 \\
\hline
G2 & 800 & -11620 & -11244 & 1800 & \textbf{-11620*} & 1800 & 0 \\
\hline 
G51 & 1000 & -3848 & -2954 & 1800 &  \textbf{-3813} & 1800 & 0.81\% \\
\hline  
G52 & 1000 & -3851 & -2958 & 1800 &  \textbf{-3826} & 1800 & 0.67\%\\
\hline
G22 & 2000 & -13359 & -9998 & 1800 & \textbf{-13046}  & 1800 & 2.35\% \\
\hline
G23 & 2000 & -13342 & - & 1800 & \textbf{-13089} & 1800 & 0.018\% \\
\hline 
G48 & 3000 & -6000 & - &  1800 & \textbf{-6000*} & 0.26 & 0 \\
\hline
G49 & 3000 & -6000 & - & 1800 & \textbf{-6000*} & 0.26 & 0 \\
\hline 
G57 & 5000 & -10000 & - & 1800 & \textbf{-10000*} & 0.51 & 0 \\
\hline
G58 & 5000 & -19293 & - & 1800 & \textbf{-18943} & 1800 & 1.81\%\\
\hline 
\end{tabular}%
\end{center}
\label{gset}
\end{table}

\section{Conclusion}

The Spin Glass Problem (SGP) is a well-known and computationally demanding problem in statistical physics, commonly tackled through combinatorial optimization. In this work, we revisited the SGP from the viewpoint of continuous nonlinear optimization, building on Rosenberg's classical result \cite{rose}, which shows that, for a broad class of multilinear polynomial problems with box constraints, the optimal values of the discrete and continuous formulations coincide.

Within this framework, we modeled the SGP as a particular instance of that class and introduced the Continuous Spin Glass (CSG) Model \eqref{model_2}, a nonconvex quadratic program with simple bound constraints. Rather than proposing new theoretical results, we provided a short argument specialized to the SGP (Proposition~1) that explains how any global minimizer of the continuous model can be transformed, in a constructive way, into an optimal $\{-1,1\}$ configuration for the SGP Model \eqref{model_1}. This argument offers a transparent and self-contained justification for the use of continuous global optimization solvers to obtain ground states, and clarifies how to deal with optimal solutions that contain non-integer components.

The computational experiments carried out on standard benchmark families support the practical relevance of this perspective. On the Biq Mac \textit{rudy} and \textit{ising} instances, the CSG Model solved by GUROBI systematically reached the known optimal value, often with substantially shorter runtimes than those required by linearization-based formulations \eqref{modelfurini2}–\eqref{modelfurini2a}. 

On the more challenging Gset instances, the CSG Model either matched the optimum or produced solutions with small relative gaps, and did so robustly for problems with up to $5000$ spins on the general-purpose CPU server described in the computational experiments. The formulation is simple to implement and can be handled by commercially available software, making it accessible even to practitioners without deep expertise in discrete optimization.

It is also useful to contrast this approach with recent ``Ising machine'' methods, which embed the Ising Hamiltonian into physical or dynamical systems and are explicitly heuristic in nature. Conversely, the models considered here remain within the mathematical programming paradigm: both the continuous and discrete formulations are solved by a global optimization solver that, when run to convergence, provides certificates of global optimality. All computations reported in this paper are carried out with such solver on general-purpose CPU hardware, rather than on dedicated Ising machines or other specialized physical devices.

Overall, our results reinforce the practical value of continuous models in addressing hard combinatorial problems such as the SGP. The CSG Model offers a simple formulation, a clear theoretical link to the discrete problem, and strong empirical performance when combined with a modern global solver, such as GUROBI. We expect that this continuous optimization viewpoint can be extended to other Ising and unconstrained binary quadratic problems, and that it may serve as a useful building block in hybrid strategies that combine continuous formulations, linearization techniques, and physics-inspired heuristics for even larger and more complex instances.

\bigskip

\section*{Acknowledgments}

The authors thank Prof. Fabio Furini and Prof. Emiliano Traversi for kindly sharing part of the computational results reported in \cite{fu19}. We are also grateful to the anonymous reviewers for their careful reading and scientifically grounded criticisms, which led to substantial improvements in the presentation, scope, and positioning of this work. Finally, we acknowledge the financial support of the Brazilian research agencies FAPESP (grant numbers 2013/07375-0, 2023/08706-1, 2024/00923-6) and CNPq (grant numbers 305227/2022-0, 404616/2024-0, 304533/2025-4, 402609/2025-5) for their financial support.


\bigskip
\bigskip


\begin{thebibliography}{00}


\bibitem{ben} U. Benlic and J. K. Hao.
Breakout local search for the max-cut problem.
\textit{Engineering Applications of Artificial Intelligence}, 26, 1162--1173, 2013.

\bibitem{bi00} P. Bialas, P. Blanchard, S. Fortunato, D. Gandolfo, and H. Satz.
Percolation and magnetization in the continuous spin Ising model.
\textit{Nuclear Physics B}, 583, 368--378, 2000.

\bibitem{bo07} E. Bolthausen and A. Bovier (Eds.).
\textit{Spin Glasses}.
Lecture Notes in Mathematics, Vol. 1900, Springer, 2007.


\bibitem{ca08} A. Caprara.
Constrained 0--1 quadratic programming: Basic approaches and extensions.
\textit{European Journal of Operational Research}, 187, 1494--1503, 2008.

\bibitem{cummins2025vectorising}
J. S. Cummins and N. G. Berloff.
Vector Ising spin annealer for minimizing Ising Hamiltonians.
\textit{Communications Physics}, 8, 225, 2025.

\bibitem{das2025classical} S. Das, S. Biswas, and B. K. Chakrabarti. Classical annealing of the Sherrington--Kirkpatrick spin glass using Suzuki--Kubo mean-field Ising dynamics.
\textit{Physical Review E}, 112, 014104, 2025.

\bibitem{du21} P. Duxbury, C. Lavor, and L. L. Salles-Neto.
A conjecture on a continuous optimization model for the Golomb Ruler Problem.
\textit{RAIRO--Operations Research}, 55, 2241--2246, 2021.

\bibitem{dra09} K. Drakakis.
A review of the available construction methods for Golomb rulers.
\textit{Advances in Mathematics of Communications}, 3, 235--250, 2009.

\bibitem{festa} P. Festa, P. Pardalos, M. G. Resende, and C. Ribeiro. Randomized heuristics for the MAX-CUT problem.
\textit{Optimization Methods and Software}, 17, 1033--1058, 2000.

\bibitem{fi91} K. Fischer and J. Hertz.
\textit{Spin Glasses}.
Cambridge University Press, 1991.


\bibitem{fu19} F. Furini and E. Traversi.
Theoretical and computational study of several linearisation techniques for binary quadratic problems.
\textit{Annals of Operations Research}, 279, 387--411, 2019.

\bibitem{glo73} F. Glover and E. Woolsey.
Further reduction of zero-one polynomial programming problems to zero-one linear programming problems.
\textit{Operations Research}, 21, 156--161, 1973.

\bibitem{glo74} F. Glover and E. Woolsey.
Converting the 0--1 polynomial programming problem to a 0--1 linear program.
\textit{Operations Research}, 22, 180--182, 1974.

\bibitem{glo75} F. Glover.
Improved linear integer programming formulations of nonlinear integer programs.
\textit{Management Science}, 22, 455--460, 1975.

\bibitem{goto2021highperformance}
H. Goto, M. Suzuki, Y. Sakai, T. Kanao, Y. Hamakawa, R. Hidaka, M. Yamasaki, and K. Tatsumura.
High-performance combinatorial optimization based on classical mechanics.
\textit{Science Advances}, 7, eabe7953, 2021.

\bibitem{gro97} W. Gropp and J. Moré.
Optimization Environments and the NEOS Server.
In \textit{Approximation Theory and Optimization}, M. D. Buhmann and A. Iserles (Eds.),
Cambridge University Press, 67--182, 1997.


\bibitem{gurobi}
Gurobi Optimization, LLC.
\textit{Gurobi Optimizer Reference Manual}, 2025.
Available at: \url{https://www.gurobi.com}.

\bibitem{ha20} F. Haydarov, S. Akhtamaliyev, M. Nazirov, and B. Qarshiyev.
Uniqueness of Gibbs measures for an Ising model with continuous spin values on a Cayley tree.
\textit{Reports on Mathematical Physics}, 86, 293--302, 2020.

\bibitem{inagaki2016coherent}
T. Inagaki, Y. Haribara, K. Igarashi, T. Sonobe, S. Tamate, T. Honjo, A. Marandi,
P. L. McMahon, T. Umeki, K. Enbutsu, O. Tadanaga, H. Takenouchi, K. Aihara,
K.-I. Kawarabayashi, K. Inoue, S. Utsunomiya, and H. Takesue.
A coherent Ising machine for 2000-node optimization problems.
\textit{Science}, 354, 603--606, 2016.

\bibitem{kri14} N. Krislock, J. Malick, and F. Roupin.
Improved semidefinite bounding procedure for solving max-cut problems to optimality.
\textit{Mathematical Programming}, 143, 61--86, 2014.


\bibitem{lucas} Lucas, A. Ising formulations of many NP problems. Frontiers in physics, 2, 5, 2014.

\bibitem{mala} S. Mallach et al.
BQP Library.
Available at: \url{http://bqp.cs.uni-bonn.de/library/html/index.html}.

\bibitem{mohseni}
N. Mohseni, P. L. McMahon, and T. Byrnes.
Ising machines as hardware solvers of combinatorial optimization problems.
\textit{Nature Reviews Physics}, 4, 363--379, 2022.

\bibitem{ouzia} H. Ouzia.
Partial reformulation-linearization based optimization models for the Golomb ruler problem.
\textit{RAIRO--Operations Research}, 58, 3171--3188, 2024.

\bibitem{re10} F. Rendl, G. Rinaldi, and A. Wiegele.
Solving max-cut to optimality by intersecting semidefinite and polyhedral relaxations.
\textit{Mathematical Programming}, 121, 307--317, 2010.

\bibitem{rose} I. Rosenberg.
Brèves communications. 0--1 optimization and non-linear programming.
\textit{Revue française d'automatique, informatique, recherche opérationnelle}, 6, 9--97, 1972.

\bibitem{she07} H. Sherali and J. Smith.
An improved linearization strategy for zero-one quadratic programming problems.
\textit{Optimization Letters}, 1, 3--47, 2007.

\bibitem{si73} B. Simon and R. Griffiths.
The $(\phi^4)_2$-field theory as a classical Ising model.
\textit{Communications in Mathematical Physics}, 33, 145--164, 1973.

\bibitem{wi07} A. Wiegele.
Biq Mac Library -- A collection of max-cut and quadratic 0--1 programming instances of medium size.
Technical Report, Alpen-Adria-Universität Klagenfurt, Austria, 2007.

\bibitem{zick} Zick, K. M. Performance report of heuristic algorithm that cracked the largest Gset Ising problems (G81 cut= 14060). arXiv preprint arXiv:2505.18508, 2025.
\end{thebibliography}
\end{document}